\documentclass{article}%
\usepackage[margin=1.55in]{geometry}
\usepackage{amsmath}
\usepackage[linkcolor=blue,urlcolor=blue,colorlinks=true]{hyperref}
\usepackage{amsfonts}
\usepackage{amssymb}
\usepackage{natbib}
\usepackage{graphicx}
\usepackage{diagbox}
\usepackage{subfig}
\usepackage{epsfig}%
\usepackage{color}
\usepackage[table]{xcolor}
\usepackage{multirow}
\usepackage{setspace}
\usepackage{hyperref}
\usepackage{enumitem}
\usepackage{hhline}
\usepackage{algorithm}
\usepackage{algorithmic}
\usepackage{pbox}
\usepackage{booktabs}
\usepackage{lscape}
\usepackage{array}
\usepackage[flushleft]{threeparttable}
\usepackage{authblk}
\usepackage{float}
\usepackage{supertabular}
\usepackage[title]{appendix}
\newtheorem{theorem}{Theorem}

\newtheorem{corollary}[theorem]{Corollary}

\newtheorem{lemma}[theorem]{Lemma}

\newenvironment{proof}[1][Proof]{\noindent\textbf{#1.} }{\ \rule{0.5em}{0.5em}}
\begin{document}
\doublespacing

\title{\textbf{A Test for Independence Via Bayesian Nonparametric Estimation of Mutual Information}}   


\author[1]{Luai Al-Labadi\thanks{{\em Corresponding author:} luai.allabadi@utoronto.ca}}

\author[2]{Forough Fazeli Asl\thanks{forough.fazeli@math.iut.ac.ir}}

\author[2]{Zahra Saberi\thanks{ z\_saberi@cc.iut.ac.ir}}

\affil[1]{Department of Mathematical and Computational Sciences, University of Toronto Mississauga, Mississauga, Ontario L5L 1C6, Canada.}
\affil[2]{Department of Mathematical Sciences, Isfahan University of Technology, Isfahan 84156-83111, Iran.}

\date{}
\maketitle

\pagestyle {myheadings} \markboth {} {A  test for independence via  BNP estimation of MI}

\begin{abstract}
Mutual information  is a well-known tool to measure the mutual dependence between variables. In this paper, a Bayesian nonparametric estimation of mutual information is established by means of the Dirichlet process  and the $k$-nearest neighbor distance. As a direct outcome of the estimation, an easy-to-implement test of independence is introduced through the relative belief ratio. Several theoretical properties of the approach are presented. The procedure is investigated through various examples where the results are compared to its frequentist counterpart and demonstrate  a good performance.
\par

 \vspace{9pt} \noindent\textsc{Keywords:} Dirichlet process, $k$-nearest neighbor distance, Mutual information, Relative belief inferences, Test for independence.

 \vspace{9pt}

\end{abstract}

\section{Introduction}
The assumption of independence is common in many fields such as statistics, data mining, machine learning and signal processing \citep{Shimizu,Fernandez,Darrell}. If this assumption is violated, the risk of having errors in the outcomes is increased. Thus, it is of particular importance to check this assumption.

A well-known tool to measure the mutual dependence between variables is the mutual information \citep{Cover}. More precisely, let $\mathbf{X}=(X_{1},\cdots,X_{d})$ be a random vector with joint continuous distribution function $F$ and marginal continuous distribution functions $F_{1},\cdots,F_{d}$. Then mutual information between $X_{1},\cdots,X_{d}$ is defined as
\begin{small}
\begin{equation}\label{MI-basic}
MI(F)=\int_{-\infty}^{\infty}\cdots\int_{-\infty}^{\infty} f(x_{1},\ldots,x_{d})\log\dfrac{f(x_{1},\ldots,x_{d})}{f(x_{1})\ldots f(x_{d})}\, dx_{1}\cdots dx_{d},
\end{equation}
\end{small}
where $f(x_{1},\cdots,x_{d})$ and $f(x_{i})$ denote, respectively, the probability density functions of $F$ and $F_{i}$, $i=1,\ldots,d$. Note that, throughout this paper, $\log(\cdot)$ denotes the natural logarithm. Clearly,  \eqref{MI-basic} is the Kullback-Leiblier of $F$ from the product of $F_i$'s and so it is non-negative. After simplification, \eqref{MI-basic} can be written as
\begin{small}
\begin{align}\label{MI-Entropy}
MI(F)&=\int_{-\infty}^{\infty}\cdots\int_{-\infty}^{\infty}f(x_{1},\ldots,x_{d})\log f(x_{1},\ldots,x_{d}) \, dx_{1}\cdots dx_{d}-\sum_{i=1}^{d}\int_{-\infty}^{\infty}f(x_{i})\log f(x_{i})\, dx_{i}\nonumber\\
&=-H(F)+\sum_{i=1}^{d}H(F_{i}),
\end{align}
\end{small}
where $H(F)$ and $H(F_{i})$ denote, respectively,  the entropy of $F$ and $F_{i}$. Accordingly, the mutual independence between $X_{1},\cdots,X_{d}$ can be tested by checking the hypothesis $\mathcal{H}_{0}: MI(F)=0$. Thus, from \eqref{MI-Entropy}, to construct a test of independence via mutual information, it is essential to develop an efficient estimator for $H(F)$ and $H(F_i)$. There have been plentiful attempts to estimate the entropy but most of them are related to the univariate (marginal) entropy estimation. See for example, \cite{Vasicek}, \cite{Ebrahimi}, \cite{Noughabi10}, \cite{Noughabi13} and \cite{Al-Omari14,Al-Omari16}. Also, \cite{Al-Labadi19d} proposed an efficient Bayesian counterpart of Vasicek's estimator. For the multivariate (joint) entropy estimation, some frequentist procedures have been offered in the literature; see, for instance, \cite{Kozachenko}, \cite{Misra}, \cite{Sricharan12}, \cite{Sricharan13},  \cite{Gao}, \cite{Berrett19a}, \cite{Ba} and the references therein. Among several estimators, due to its simplicity, \cite{Kozachenko} (KL) estimator is the most common one. Let $\mathbf{X}_{1},\cdots,\mathbf{X}_{n}$ be $n$ independent random vectors each having the continuous $d$-variate cdf $F$ and let, for $i=1,\ldots,n$, $\rho_{i}=\min\lbrace||\mathbf{X}_{i}-\mathbf{X}_{j}||,\, j\in\lbrace 1,\ldots,n\rbrace\setminus\lbrace i\rbrace\rbrace$, where $||\cdot||$ denotes the Euclidean norm on $\mathbb{R}^{d}$ and $A\setminus B$ denotes the set of elements in $A$ but not in $B$. Then, the KL estimator is given by
\begin{small}
\begin{equation}\label{KL-estimator}
H^{KL}_{n}=\frac{d}{n}\sum_{i=1}^{n}\log \rho_{i}+\log \left( \dfrac{\pi^{\frac{d}{2}}}{\Gamma(\frac{d}{2}+1)}\right)+\gamma+\log (n-1),
\end{equation}
\end{small}
where $\gamma=0.5772\cdots$ denotes Euler's constant. \cite{Kozachenko} showed that $H^{KL}_{n}$ is a consistent estimator. However, \cite{Singh} remarked that, in practical applications, the estimator \eqref{KL-estimator} can be applied when the small values of the nearest neighbor distance $\rho_{i}$'s are recorded to high accuracy, which is often not the case. They improved the estimator $H^{KL}_{n}$ in \eqref{KL-estimator} by proposing the following  $k$-nearest neighbor ($k$-NN) version of KL estimator:
\begin{small}
\begin{equation}\label{k-NN.KL}
H^{k.KL}_{n}=\frac{d}{n}\sum_{i=1}^{n}\log R_{i,k,n-1}+\log \left( \dfrac{\pi^{\frac{d}{2}}}{\Gamma(\frac{d}{2}+1)}\right)-L_{k-1}+\gamma+\log n,
\end{equation}
\end{small}
where, $L_{0}=0$, $L_{j}=\sum_{r=1}^{j}\frac{1}{r}$, $R_{i,k,n-1}=||\mathbf{X}_{(k),i}-\mathbf{X}_{i}||$ and $\mathbf{X}_{(1),i},\ldots,\mathbf{X}_{(k),i},\ldots,$ $\mathbf{X}_{(n-1),i}$ is a reordering of $\lbrace \mathbf{X}_{1},\ldots,\mathbf{X}_{n}\rbrace\setminus\lbrace\mathbf{X}_{i}\rbrace$ such that $||\mathbf{X}_{(1),i}-\mathbf{X}_{i}||\leq\ldots\leq||\mathbf{X}_{(k),i}-\mathbf{X}_{i}||\leq\ldots\leq||\mathbf{X}_{(n-1),i}-\mathbf{X}_{i}||$. \cite{Singh} proved the asymptotic unbiasedness and consistency of $H^{k.KL}_{n}$.
They used Monte Carlo simulations to find a suitable choice of $k$. For instance, they recommended using $k=4$ as an optimal choice  for sample sizes $n\leq50$.

A primary application of mutual information is  to build tests of independence. For example, in a recent work, \cite{Berrett19b}  developed a test of independence based on a weighted version of the KL estimator. Additional  tests of independence that count on mutual information can be found in \cite{Wu}, \cite{Mathew} and \cite{Pethel}. For other strategies of tests of independence such as copula process, distance covariance, and etc; see, \cite{Genest}, \cite{Kojadinovic}, \cite{Medovikov}, \cite{Belalia}, \cite{Susam}, \cite{Karvanen}, \cite{Meintanis}, \cite{GaiBer}, and \cite{Fan} for a comprehensive review.  \citet[pp. 12-15]{Roy} pointed out that most of these methods suffer from a weak performance for sample sizes less than or equal to 50.

As seen earlier, there are extensive frequentist multivariate entropy estimations. On the other side, Bayesian estimation has been not received much attention. To the best knowledge of the authors, there are only two works related to test of independence that use Bayesian nonparametric (BNP) techniques. The first one, due to \cite{Filippi16}, uses Dirichlet process mixture prior on the unknown distribution of the data to present two BNP diagnostic measures for detecting pairwise dependencies that are scalable to large data sets. The second work, due to \cite{Filippi17}, considers P\'{o}lya tree prior to derive an explicit form of Bayes factor to state evidence for independence between pairs of random variables. Both of the previous works do not rely on entropy estimation. Thus,  deriving a general BNP  estimator of entropy that supports both marginal and joint entropy estimation with small systematic errors appears thought-provoking.  Developing such an estimator will be the first goal of this paper. Having the estimator in hand makes it possible to construct a Bayesian test for independence. The anticipated estimator may be viewed as the BNP counterpart of  \eqref{k-NN.KL}. The Dirichlet process and relative belief ratio are utilized to build the test. As seen in the next sections, the developed test is easy-to-implement with an excellent performance particularly  for small sample sizes.


The reminder of this paper is as follow. A relevant background containing some definitions and generic properties of the Dirichlet process and the relative belief ratio are reviewed in Section \ref{sec-back}. Section \ref{sec-estimation} is a central section where  a BNP estimator of mutual information is developed through estimating joint and marginal entropies. In addition, several theoretical properties of the proposed estimator are derived. It also discusses  the choice of the hyperparameter of the Dirichlet process. In Section \ref{sec-pri.test}, a test for independence is presented as a result of the estimation of  mutual information. Computational algorithms to implement the approach are outlined in Section \ref{sec-computations}. In Section \ref{simulation}, the procedure is investigated through several examples where the results are compared to its frequentist counterpart. Finally, Section \ref{conclusion} concludes the paper with a summary of the results. A short proof to clarify some expressions related to Section \ref{sec-estimation} is given in the Appendix.

\section{Relevant Background}\label{sec-back}
\subsection{Dirichlet Process}\label{DP-sub}
The Dirichlet process, introduced by \cite{Ferguson}, is the
most commonly used prior in BNP inferences. A remarkable collection of nonparametric inferences have been devoted to this prior. In this section, we only present the most relevant
definitions and properties of this prior. Consider a space
$\mathfrak{X}$ with a $\sigma$-algebra $\mathcal{A}$ of subsets of
$\mathfrak{X}$, let $G$ be a fixed probability measure on $(\mathfrak{X}%
,\mathcal{A}),$ called the \emph{base measure}, and $a$ be a positive number,
called the \emph{concentration parameter}. A random probability measure
$P=\left\{  P(A):A\in\mathcal{A}\right\}  $ is called a Dirichlet process on
$(\mathfrak{X},\mathcal{A})$ with parameters $a$ and $G,$ denoted by $P\sim
{DP}(a,G),$ if for every measurable partition $A_{1},\ldots,A_{k}$ of
$\mathfrak{X} $ with $k\geq2\mathfrak{,}$ the joint distribution of the vector
$\left(  P(A_{1}),\ldots\,P(A_{k})\right)$ has the Dirichlet distribution with parameter
$aG(A_{1}),\ldots,$ $aG(A_{k})$. Also, it is assumed that
$G(A_{j})=0$ implies $P(A_{j})=0$ with probability one. Consequently, for any
$A\in\mathcal{A}$, $P(A)\sim$ beta$(aG(A),a(1-G(A)))$,
${E}(P(A))=G(A)\ $and ${Var}(P(A))=G(A)(1-G(A))/(1+a).$ Accordingly, the base measure $G$
plays the role of the center of $P$ while the concentration parameter $a$ controls the variation of
 $P$ around  $G$. One of the most well-known
properties of the Dirichlet process is the conjugacy property. That is, when the sample $x=(x_{1},\ldots,x_{n})$
is drawn from $P\sim DP(a,G)$, the posterior distribution of $P$ given $x$,
denoted by $P^{\ast}$, is also a Dirichlet process with
concentration parameter $a+n$ and base measure
\begin{equation}\label{pos base measure}
G^{\ast}_{a,n}=a(a+n)^{-1}G+n(a+n)^{-1}F_{n},
\end{equation}
where $F_{n}$ denotes the empirical cumulative distribution function (cdf) of the sample
$x$. Note that, $G^{\ast}_{a,n}$ is a convex combination
of the base measure $G$ and the empirical cdf $F_{n}$. Therefore, $G^{\ast}_{a,n}\rightarrow G$ as
$a\rightarrow\infty$ while $G^{\ast}_{a,n}\rightarrow F_{n}$ as $a\rightarrow0$. On the other hand, by
Glivenko-Cantelli theorem, when $n\rightarrow\infty$, $G^{\ast}_{a,n}$ converges to true distribution function generating the data. A guideline
about choosing the hyperparameters $a$ and $G$ will be covered for the  test of independence in Section \ref{sec-pri.test}. Following \cite{Ferguson}, $P\sim{DP}(a,G)\ $ can be represented as
\begin{equation}
P=\sum_{i=1}^{\infty}L^{-1}(\Gamma_{i}){\delta_{Y_{i}}/}\sum_{i=1}^{\infty
}{{L^{-1}(\Gamma_{i})}}, \label{series-dp}%
\end{equation}
where $\Gamma_{i}=E_{1}+\cdots+E_{i}$ with $E_{i}\overset{i.i.d.}{\sim}%
$\ exponential$(1),Y_{i}\overset{i.i.d.}{\sim}G$ independent of the
$\Gamma_{i},L^{-1}(y)=\inf\{x>0:L(x)\geq y\}$ with $L(x)=a\int_{x}^{\infty
}t^{-1}e^{-t}dt,x>0,$ and ${\delta_{a}}$ the Dirac delta measure. The series representation
(\ref{series-dp}) implies that the Dirichlet process is a discrete probability
measure even for the cases with an absolutely continuous base measure $G$. Note
that, by imposing the weak topology, the support of the Dirichlet process could
be quite large, namely, the support is the set of all probability measures whose support is contained in the support of the base measure. Recognizing the complexity when working with \eqref{series-dp} (i.e., no closed form for the inverse of L\'{e}vy measure $L(x)$ exists), \cite{Ishwaran} proposed the following finite representation as an efficient method to simulate the Dirichlet process. They showed that the Dirichlet process $P\sim DP(a,G)$ can be approximated by
\begin{equation}\label{approx of DP}
P_{N}=\sum_{i=1}^{N}J_{i,N}\delta_{Y_{i}},
\end{equation}
where, $(J_{1,N},\ldots, J_{N,N})\sim$ Dirichlet$(a/N,\ldots,a/N)$. Then $E_{P_{N}}(g)\rightarrow E_{P}(g)$ in distribution as $N\rightarrow\infty$, for any measurable function $g: \mathbb{R}\rightarrow \mathbb{R}$ with $\int_{\mathbb{R}}|g(x)|\, H(dx)<\infty$ and $P\sim DP(a,H)$. In particular, $(P_{N})_{N\geq 1}$ converges in distribution to $P$, where $P_{N}$ and $P$ are random values in the space $M_{1}(\mathbb{R})$ of probability measures on $\mathbb{R}$ endowed with the topology of weak convergence. To generate $(J_{i,N})_{1\leq i\leq N}$ put $J_{i,N}=G_{i,N}/\sum_{i=1}^{N}G_{i,N}$, where $(G_{i,N})_{1\leq i\leq N}$ is a sequence of i.i.d. gamma$(a/N, 1)$ random variables independent of $(Y_{i})_{1\leq i\leq N}$. This form of approximation leads to some results in Section \ref{sec-estimation}.

\subsection{Relative Belief Inferences}\label{RB-sub}
The relative belief ratio, developed by \cite{Evans15}, becomes a widespread measure of statistical evidence. See, for example, the work of \cite{Al-Labadi18}, \cite{Al-Labadi-E17,Al-Labadi-E18}, \cite{Al-Labadi19a,Al-Labadi19b} and \cite{Al-Labadi19c} for implementation of the relative belief ratio on different stimulating model checking problems. In details, let $\{f_{\theta}:\theta\in\Theta\}$ be a collection of densities on a sample space $\mathfrak{X}$ and let $\pi$ be a prior on the parameter space $\Theta$. Note that
the densities may represent discrete or continuous probability measures but they are
all with respect to the same support measure  $d\theta$. After
observing the data $x,$ the posterior distribution of $\theta$, denoted by $\pi(\theta\,|\,x)$, is a revised prior and is given by the
density $\pi(\theta\,|\,x)=\pi(\theta)f_{\theta}(x)/m(x)$, where $m(x)=\int
_{\Theta}\pi(\theta)f_{\theta}(x)\,d\theta$ is the prior predictive density of
$x.$  For a parameter of interest $\psi=\Psi(\theta),$ let $\Pi_{\Psi}$ be
the marginal prior probability measure and $\Pi_{\Psi}(\cdot|\,x)$ be
the marginal posterior probability measure. It is assumed that
$\Psi$ satisfies regularity conditions
so that the prior density $\pi_{\Psi}$ and the posterior density
$\pi_{\Psi}(\cdot\,|\,x)$ of $\psi$ exist with respect to some support measure on the range space for $\Psi$
. The relative belief ratio for a value
$\psi$ is then defined by $RB_{\Psi}(\psi\,|\,x)=\lim_{\delta\rightarrow0}%
\Pi_{\Psi}(N_{\delta}(\psi\,)|\,x)/\Pi_{\Psi}(N_{\delta}(\psi\,)),$ where
$N_{\delta}(\psi\,)$ is a sequence of neighborhoods of $\psi$ converging
nicely to $\psi$ as $\delta\rightarrow0$ \citep{Evans15}. When $\pi_{\Psi}$ and  $\pi_{\Psi}(\cdot\,|\,x)$ are continuous at $\psi,$ the relative belief ratio is defined by
\begin{equation*}
RB_{\Psi}(\psi\,|\,x)=\pi_{\Psi}(\psi\,|\,x)/\pi_{\Psi}(\psi), \label{relbel}%
\end{equation*}
the ratio of the posterior density to the prior density at $\psi.$  Therefore,
$RB_{\Psi}(\psi\,|\,x)$ measures the change in the belief of $\psi$ being the true value from a \textit{priori} to a \textit{posteriori}.

Since $RB_{\Psi}(\psi\,|\,x)$ is a measure of the evidence that $\psi$ is the true value, if $RB_{\Psi}(\psi\,|\,x)$ $>1$, then the probability of $\psi$ being the true value from a priori to a posteriori is increased, consequently there is evidence based on the data that $\psi$ is the true value. If $RB_{\Psi}(\psi\,|\,x)<1$, then the probability of  $\psi$ being the true value from a priori to a posteriori is decreased. Accordingly, there is evidence against based on the data that $\psi$ being the true value. For the case $RB_{\Psi}(\psi\,|\,x)=1$ there is no
evidence either way.

Obviously, $RB_{\Psi}(\psi_{0}\,|\,x)$ measures the evidence of the hypothesis $\mathcal{H}_{0}:\Psi(\theta)=\psi_{0}$. Large values of $RB_{\Psi}(\psi_{0}\,|\,x)=c$ provides
 strong evidence in favor of $\psi_{0}$. However, there may also exist other
values of $\psi$ that had even larger increases. Thus, it is also necessary, however, to calibrate whether this is strong or weak evidence for
or against $\mathcal{H}_{0}.$ A typical
calibration of $RB_{\Psi}(\psi_{0}\,|\,x)$  is given by the  \textit{strength}
\begin{equation}
\Pi_{\Psi}\left[RB_{\Psi}(\psi\,|\,x)\leq RB_{\Psi}(\psi_{0}\,|\,x)\,|\,x\right].
\label{strength}%
\end{equation}
The value in \eqref{strength} indicates that the posterior probability that the true value of $\psi$ has a relative
belief ratio no greater than that of the hypothesized value $\psi_{0}.$ Noticeably, (\ref{strength}) is not a p-value as it has a very different
interpretation. When $RB_{\Psi}(\psi_{0}\,|\,x)<1$, there is evidence
against $\psi_{0},$ then a small value of (\ref{strength}) indicates
 strong evidence against $\psi_{0}$. On the other hand, a large value for \eqref{strength}    indicates   weak evidence against $\psi_{0}$.
Similarly, when $RB_{\Psi}(\psi_{0}\,|\,x)>1$, there is  evidence in favor
of $\psi_{0},$ then a small value of (\ref{strength}) indicates  weak
evidence in favor of $\psi_{0}$, while a large value of \eqref{strength} indicates
 strong evidence in favor of $\psi_{0}$.

\section{BNP Posterior of mutual information}\label{sec-estimation}
In this section, we provide a posterior of entropy and use it in \eqref{MI-Entropy} to propose a posterior of mutual information.
\subsection{Prior and Posterior of Entropy  }\label{sec.Est-sub1}
Let $P_{N}=\sum_{i=1}^{N}J_{i,N}\delta_{Y_{i}}$ be as defined by \eqref{approx of DP}, where $(J_{1,N},\ldots,J_{N,N})\sim Dirichlet(a/N,$ $\ldots,a/N)$, $Y_{1},\ldots,Y_{N}\overset{i.i.d.}{\sim}G$, and $Y_{i}\in \mathbb{R}^{d}$. The proposed $k$-NN BNP prior of entropy is defined by
\begin{small}
\begin{align}\label{k-NN.pri}
H^{pri}_{N,a,k}&=\sum_{i=1}^{N}J_{i,N}\left(\log \frac{(N-1)\pi^{\frac{d}{2}}R^{d}_{i,k,N-1}}{k\Gamma(\frac{d}{2}+1)}\right)-L_{k-1}+\gamma+\log k\nonumber\\
&=\sum_{i=1}^{N}J_{i,N}T_{i}^{(N-1)}-L_{k-1}+\gamma+\log k,
\end{align}
\end{small}
where $k\in\lbrace 1,\ldots,N-1\rbrace$ and $R_{i,k,N-1}$ is the euclidean distance between $Y_{i}$ and its $k$-th closest neighbor. The next lemma shows the asymptotic behavior of the expectation and the variance of $H^{pri}_{N,a,k}$, when $N\rightarrow\infty $ and $a\rightarrow\infty$.
\begin{lemma}\label{E-V.pri}
Let $G$ be a $d$-variate distribution and $F\sim DP(a,G)$. Consider the $k$-NN BNP prior $H^{pri}_{N,a,k}$ as defined in \eqref{k-NN.pri}, then
\begin{itemize}
\item[i.]
$E(H^{pri}_{N,a,k})\rightarrow H(G)$, as $N\rightarrow\infty$,
\item[ii.]
$Var(H^{pri}_{N,a,k})\rightarrow 0$, as $N\rightarrow\infty$ and $a\rightarrow\infty$.
\end{itemize}
\end{lemma}
\begin{proof}
To prove (i), since $J_{i,N}$ and $T_{i}^{(N-1)}$ are independent, we have
\begin{small}
\begin{align*}
E(H^{pri}_{N,a,k})&=\sum_{i=1}^{N}\left(E(T_{i}^{(N-1)})E\left(J_{i,N}\right)\right)-L_{k-1}+\gamma+\log k,
\end{align*}
\end{small}
Noting that $E(J_{i,N})=1/N$, and $(T_{i}^{(N-1)})_{1\leq i\leq N}$ are identically distributed random variables, we have
\begin{small}
\begin{align*}
E(H^{pri}_{N,a,k})&=E(T_{1}^{(N-1)})-L_{k-1}+\gamma+\log k.
\end{align*}
\end{small}
From \cite{Singh}, $E(T_{1}^{(N-1)})\rightarrow L_{k-1}-\gamma-\log k+H(G)$ as $N\rightarrow\infty$, and the result follows. To prove (ii), since $Var(J_{i,N})=\frac{N-1}{N^{2}(a+1)}$ and $Cov(J_{i,N},J_{j,N})=-\frac{1}{N^{2}(a+1)}$, we have
\begin{small}
\begin{align}\label{Var.pri1}
Var(H^{pri}_{N,a,k})&=Var\left(E\left(\sum_{i=1}^{N}T_{i}^{(N-1)}J_{i,N}|T_{1}^{(N-1)},\ldots,T_{N}^{(N-1)}\right)\right)\nonumber\\
&+E\left(Var\left(\sum_{i=1}^{N}T_{i}^{(N-1)}J_{i,N}|T_{1}^{(N-1)},\ldots,T_{N}^{(N-1)}\right)\right)\nonumber\\
&=Var\left(\frac{1}{N}\sum_{i=1}^{N}T_{i}^{(N-1)}\right)
+E\Bigg(\frac{N-1}{N^{2}(a+1)}\sum_{i=1}^{N}\left(T_{i}^{(N-1)}\right)^{2}
-\frac{2}{N^{2}(a+1)}\nonumber\\
&\times\sum_{i<j}^{N}T_{i}^{(N-1)}T_{j}^{(N-1)}\Bigg)\nonumber\\
&=\frac{1}{N}Var\left(T_{1}^{(N-1)}\right)+\frac{N-1}{N}Cov\left(T_{1}^{(N-1)},T_{2}^{(N-1)}\right)+\frac{N-1}{N(a+1)}\Bigg(E\left(T^{(N-1)}_{1}\right)^{2}\nonumber\\
&-E\left(T_{1}^{N-1}T_{2}^{N-1}\right)\Bigg).
\end{align}
\end{small}
From \cite{Singh}, $Var\left(T_{1}^{(N-1)}\right)\rightarrow Q_{k}+Var\left(\log g(\mathbf{y})\right)$, $Cov\left(T_{1}^{(N-1)},T_{2}^{(N-1)}\right)$ $\rightarrow 0$, $E\left(T^{(N-1)}_{1}\right)^{2}\rightarrow Q_{k}+Var\left(\log g(\mathbf{y})\right)+[L_{K-1}-\gamma-\log k+H(G)]^{2}$, and $E(T_{1}^{N-1}T_{2}^{N-1})$ $\rightarrow[L_{K-1}-\gamma-\log k+H(G)]^{2}$ as $N\rightarrow\infty$, where $Q_{k}=\sum_{j=k}^{\infty}\frac{1}{j^{2}}$, $g$ denotes the probability density function of $G$ and $\mathbf{y}\in\mathbb{R}^{d}$. Hence, by letting $N\rightarrow\infty$ in \eqref{Var.pri1}, we have
\begin{small}
\begin{align}\label{Var.pri2}
Var(H^{pri}_{N,a,k})\rightarrow\frac{1}{a+1}\left\lbrace Q_{k}+Var\left(\log g(\mathbf{y})\right)\right\rbrace.
\end{align}
\end{small}
Letting $a\rightarrow\infty$, gives the proof of (ii).
\end{proof}

The next corollary shows the asymptotic behavior of the variance of $H^{pri}_{N,a,k}$ when $N\rightarrow\infty$ and $k\rightarrow\infty$.
\begin{corollary}\label{corollary}
Consider $H^{pri}_{N,a,k}$ as defined in \eqref{k-NN.pri}. Then, for fixed $a$, as $N\rightarrow\infty$ and $k\rightarrow\infty$, we have
\begin{small}
\begin{align*}
Var(H^{pri}_{N,a,k})\rightarrow\frac{1}{a+1}Var\left(\log g(\mathbf{y})\right).
\end{align*}
\end{small}
\end{corollary}
\begin{proof}
Note that $Q_{k}=\sum_{j=k}^{\infty}\frac{1}{j^{2}}$ can be written as $\int_{0}^{\infty}\frac{t}{1-e^{-t}}e^{-kt}\, dt$ \citep[p. 260]{Abramowitz}. Hence, by letting $k\rightarrow\infty$ in \eqref{Var.pri2}, the monotone convergence theorem implies that $Q_{k}\rightarrow 0$ and the result follows.
\end{proof}

From Corollary \ref{corollary}, for a fixed value of $a$, choosing too large values of $k$ reduces the statistical errors; however, in practical applications, for such values of $k$, the increase of systematic errors outweighs the decrease of statistical errors \citep{Singh,Kraskov}. In Section \ref{simulation}, we performed a simulation study to assess the effect of different values of $k$ on the behavior of the systematic errors. As a result, we recommend choosing $k=3$ in the BNP procedure.

Now, by the conjugacy property of the Dirichlet process, the BNP posterior of entropy  can be proposed as follows.
Assume that $\mathbf{x}_{d\times n}=(\mathbf{x}_{1},\ldots,\mathbf{x}_{n})$ is an observed sample of size $n$ from an unknown $d$-variate distribution $F$, where $\mathbf{x}_{i}\in\mathbb{R}^{d}$, $i=1,\ldots,n$. Note that, the subscript $d\times n$ may be omitted whenever it is clear in the context. To present the BNP posterior of entropy, we use the prior $F\sim DP(a,G)$ for some choices of $a$ and $d$-variate distribution $G$. By \eqref{pos base measure}, $F^{\ast}:=F|\mathbf{x}\sim DP(a+n,G^{\ast}_{a,n})$. The BNP  posterior of entropy is proposed by
\begin{small}
\begin{align}\label{k-NN.post}
H^{pos}_{N,a+n,k}&=\sum_{i=1}^{N}J^{\ast}_{i,N}\left(\log \frac{(N-1)\pi^{\frac{d}{2}}(R^{\ast}_{i,k,N})^{d}}{k\Gamma(\frac{d}{2}+1)}\right)-L_{k-1}+\gamma+\log k,
\end{align}
\end{small}
where $(J^{\ast}_{i,N})_{1\leq i\leq N}\sim Dirichlet((a+n)/N,\ldots,(a+n)/N)$, $Y^{\ast}_{1},\ldots,Y^{\ast}_{N}\overset{i.i.d.}{\sim}G^{\ast}_{a,n}$ and $k\in\lbrace1,\ldots,N-1\rbrace$. In the same manner, for $i= 1,\ldots,d$, the marginal entropy $H(F_{i})$ can be estimated by using prior $F_{i}\sim DP(a,G_{i})$, where $G_{i}$ is the $i$-th marginal of the cdf $G$. The convergence of $E(H^{pos}_{N,a+n,k})$ to the entropy of the true distribution will be shown in the next theorem. As we will show later, the entropy of $G^{\ast}_{a,n}$ has a crucial role in the proof of this convergence. To carry on, some notations and theoretical results related to $H(G^{\ast}_{a,n})$ are first presented.

Let $F_{1},\ldots,F_{m}$ be $m$ cdf's defined on the same probability space and $F_{\boldsymbol\alpha}=\sum_{i=1}^{m}\alpha_{i}F_{i}$ so that $\sum_{i=1}^{m}\alpha_{i}=1$. The following result due to \citet[p. 4226]{Toomaj} gives the entropy of $F_{\boldsymbol\alpha}$. Let $D_{kull}(F_{i},F_{\boldsymbol\alpha})$ denote the Kullback-Leibler divergence between $F_{i}$ and $F_{\boldsymbol\alpha}$, $i=1,\ldots,m$, then
\begin{small}
\begin{align}\label{mix-E}
H(F_{\boldsymbol\alpha})=\sum_{i=1}^{m}\alpha_{i}H(F_{i})+\sum_{i=1}^{m}\alpha_{i}D_{kull}(F_{i},F_{\boldsymbol\alpha}).
\end{align}
\end{small}
Now, by applying \eqref{mix-E} for \eqref{pos base measure}, we have
\begin{small}
\begin{align}\label{mix-G}
H(G^{\ast}_{a,n})&=\frac{a}{a+n}H(G)+\frac{n}{a+n}H(F_{n})+\frac{n}{a+n}D_{kull}(F_{n},G^{\ast}_{a,n})+\frac{a}{a+n}D_{kull}(G,G^{\ast}_{a,n}).
\end{align}
\end{small}
Note that, $D_{kull}(\cdot,\cdot)$ is only defined for two cdf's on the same probability space (both cdf's should be continuous or discrete). Since $G^{\ast}_{a,n}$ is not completely continuous or discrete, $D_{kull}(G,G^{\ast}_{a,n})$ and $D_{kull}(F_{n},G^{\ast}_{a,n})$ in \eqref{mix-G} do not make sense. To avoid this difficulty, we  define   $D_{kull}(G,G^{\ast}_{a,n})$ and $D_{kull}(F_n,G^{\ast}_{a,n})$ by encoding the distributions $G$, $F_n$ and $G^{\ast}_{a,n}$ around a set of the $d$-dimensional real valued points through the next lemma. In fact, we use a method of discretization to define $G$, $F_{n}$ and $G^{\ast}_{a,n}$ on a same probability space.
\begin{lemma}\label{def}
Consider $G$, $F_{n}$ and $G^{\ast}_{a,n}$ as defined in \eqref{pos base measure}. Let $\mathcal{I}\subseteq\mathbb{N}$ and $\lbrace\mathbf{t}_{j}\rbrace_{j\in\mathcal{I}}\subseteq\mathbb{R}^{d}$ be such that for a given $\delta>0$
\begin{small}
\begin{align}
g_{j}&=Pr\left(t_{j1}-\delta<Z_{1}\leq t_{j1},\ldots,t_{jd}-\delta<Z_{d}\leq t_{jd}\right)\nonumber\\\label{g}
&=G(t_{j1},\ldots,t_{jd})+(2^{d}-3)G(t_{j1}-\delta,\ldots,t_{jd}-\delta)-\underset{S}{\sum} G(s_{1},\ldots,s_{d}),
\end{align}
\end{small}
\begin{small}
\begin{align}
f_{j,n}&=Pr\left(t_{j1}-\delta<Z^{\prime}_{1}\leq t_{j1},\ldots,t_{jd}-\delta<Z^{\prime}_{d}\leq t_{jd}\right)\nonumber\\\label{fn}
&=F_{n}(t_{j1},\ldots,t_{jd})+(2^{d}-3)F_{n}(t_{j1}-\delta,\ldots,t_{jd}-\delta)-\underset{S}{\sum} F_{n}(s_{1},\ldots,s_{d}),\end{align}
\end{small}
and
\begin{small}
\begin{align}
g^{\ast}_{j,a,n}&=Pr\left(t_{j1}-\delta<Z^{\prime\prime}_{1}\leq t_{j1},\ldots,t_{jd}-\delta<Z^{\prime\prime}_{d}\leq t_{jd}\right)\nonumber\\\label{g-ast}
&=G^{\ast}_{a,n}(t_{j1},\ldots,t_{jd})+(2^{d}-3)G^{\ast}_{a,n}(t_{j1}-\delta,\ldots,t_{jd}-\delta)-\underset{S_d}{\sum} G^{\ast}_{a,n}(s_{1},\ldots,s_{d}),\end{align}
\end{small}
satisfy conditions $g_{j}=0$ and $f_{j,n}=0$ whenever $g^{\ast}_{j,a,n}=0$, $\sum_{j\in\mathcal{I}}g^{\ast}_{j,a,n}\leq\sum_{j\in\mathcal{I}}g_{j}\leq1$, $\sum_{j\in\mathcal{I}}g^{\ast}_{j,a,n}\leq\sum_{j\in\mathcal{I}}f_{j,n}\leq1$, where $\mathbf{Z}\sim G$, $\mathbf{Z}^{\prime}\sim F_{n}$, $\mathbf{Z}^{\prime\prime}\sim G^{\ast}_{a,n}$ and $S_d=\lbrace (s_{1},\ldots,s_{d}):s_{k}\in\lbrace t_{jk}-\delta,t_{jk}\rbrace, k\in\lbrace1\ldots d\rbrace\rbrace\setminus \lbrace(t_{j1},\ldots,t_{jd}),$ $(t_{j1}-\delta,\ldots,t_{jd}-\delta)\rbrace$. $D_{kull}(G,G^{\ast}_{a,n})$ and $D_{kull}(F_{n},G^{\ast}_{a,n})$, respectively, can be (empirically) defined as $\sum_{j\in\mathcal{I}}\big(g_{j}$ $\log\frac{g_{j}}{g^{\ast}_{j,a,n}}\big)$ and $\sum_{j\in\mathcal{I}}\big(f_{j,n}\log\frac{f_{j,n}}{g^{\ast}_{j,a,n}}\big)$ by applying the general definition of the Kullback-Leibler \citep[p. 34]{MacKay} based on atoms $g_{j}$, $f_{j,n}$ and $g^{\ast}_{j,a,n}$ with the standard convention $0\log\frac{0}{0}=0$ \citep[p. 91]{Piera}.
\end{lemma}
\begin{proof}
An inductive procedure to derive \eqref{g}, \eqref{fn} and \eqref{g-ast} is given by Appendix A.
\end{proof}

Note that defining $g_{j}$, $f_{j,n}$ and $g^{\ast}_{j,a,n}$, respectively, based on $G$, $F_n$ and $G^{\ast}_{a,n}$ play a key role in the proof of the next theorem.
\begin{theorem}\label{pos-mean}
Let $\mathbf{x}_{d\times n}$ be a sample from $d$-variate distribution function $F$ and $F|\mathbf{x}\sim DP(a+n,G^{\ast}_{a,n})$. Assume that the limit of $D_{kull}(G,G^{\ast}_{a,n})$ exists, as $n\rightarrow\infty$.
Then, $E(H^{pos}_{N,a+n,k})\rightarrow H(F)$, as $N\rightarrow\infty$ and $n\rightarrow\infty$.
\end{theorem}
\begin{proof}
From the conjugacy property of the Dirichlet process, part (i) of Lemma \eqref{E-V.pri} implies that $E(H^{pos}_{N,a+n,k})\rightarrow H(G^{\ast}_{a,n})$, as $N\rightarrow\infty$. Consider $H(G^{\ast}_{a,n})$ as defined in \eqref{mix-G}. Then, $\frac{a}{a+n}H(G)\rightarrow0$ as $n\rightarrow\infty$. Also, the strong law of large numbers implies that $H(F_{n})=-n^{-1}\sum_{i=1}^{n}\log(f(\mathbf{x}_{i}))\rightarrow H(F)$ as $n\rightarrow\infty$. Now, consider $D_{kull}(F_{n},G^{\ast}_{a,n})$ as given by Lemma \ref{def}. From \eqref{pos base measure} and \eqref{g-ast}, we get
\begin{small}
\begin{align*}
g^{\ast}_{j,a,n}&=\frac{a}{a+n}G(t_{j1},\ldots,t_{jd})+\frac{n}{a+n}F_{n}(t_{j1},\ldots,t_{jd})+(2^{d}-3)\bigg\lbrace\frac{a}{a+n}G(t_{j1}-\delta,\ldots,t_{jd}-\delta)\nonumber\\
&+\frac{n}{a+n}F_{n}(t_{j1}-\delta,\ldots,t_{jd}-\delta)\bigg\rbrace-\underset{S}{\sum} \left(\frac{a}{a+n}G(s_{1},\ldots,s_{d})+\frac{n}{a+n}F_{n}(s_{1},\ldots,s_{d})\right).
\end{align*}
\end{small}
After some simplification, we have
\begin{small}
\begin{align*}
g^{\ast}_{j,a,n}&=\frac{a}{a+n}\bigg\lbrace G(t_{j1},\ldots,t_{jd})+(2^{d}-3)G(t_{j1}-\delta,\ldots,t_{jd}-\delta)-\underset{S}{\sum}G(s_{1},\ldots,s_{d})\bigg\rbrace\nonumber\\
&+\frac{n}{a+n}\bigg\lbrace F_{n}(t_{j1},\ldots,t_{jd})+(2^{d}-3)F_{n}(t_{j1}-\delta,\ldots,t_{jd}-\delta)-\underset{S}{\sum}F_{n}(s_{1},\ldots,s_{d})\bigg\rbrace.
\end{align*}
\end{small}
Now, using \eqref{g} and \eqref{fn}, we have $g^{\ast}_{j,a,n}=\frac{a}{a+n}g_{j}+\frac{n}{a+n}f_{j,n}$. Clearly, $g^{\ast}_{j,a,n}\geq\frac{n}{a+n}f_{j,n}$ which concludes that $\frac{f_{j,n}}{g^{\ast}_{j,a,n}}\leq 1+a/n\leq 1+a$, for $j\in\mathcal{I}$. Consequently, $D_{kull}(F_{n},G^{\ast}_{a,n})\leq\log(1+a)<\infty$. On the other hand, applying the Glivenko-Cantelli theorem in \eqref{fn} and \eqref{g-ast} implies $f_{j,n}\xrightarrow{a.s.} f_j$ and $g^{\ast}_{j,a,n}\xrightarrow{a.s.} f_j$, as $n\rightarrow\infty$, where $f_{j}$ denotes $F(t_{j1},\ldots,t_{jd})+F(t_{j1}-\delta,\ldots,t_{jd}-\delta)-\underset{S}{\sum} F(s_{1},\ldots,s_{d})$. Hence, by the discrete version of the dominated convergence theorem, we have
\begin{small}
\begin{align*}
D_{kull}(F_{n},G^{\ast}_{a,n})=\sum_{j\in\mathcal{I}}f_{j,n}(\log f_{j,n}-\log g^{\ast}_{j,a,n})\xrightarrow{a.s.} \sum_{j\in\mathcal{I}}f_{j}(\log f_{j}-\log f_{j})=0.
\end{align*}
\end{small}
The proof is completed by letting $n\rightarrow\infty$ in the last term of \eqref{mix-G}.
\end{proof}
\subsection{Posterior of Mutual Information}\label{sec.Est-sub2}
The proposed BNP posterior of mutual information takes the form:
\begin{small}
\begin{align*}
MI^{pos}=\left[-H^{pos}_{N,a+n,k}(F)+\sum_{i=1}^{d}H^{pos}_{N,a+n,k}(F_{i})\right]^{+},
\end{align*}
\end{small}
where $b^{+}=\max(b,0)$. Note that, the proposed estimator ensures the non-negativity of the BNP mutual information estimation. The BNP test in the next section will be proposed based on $MI^{pos}$. Since
 implementation of $MI^{pos}$ requires considering choices of $a$ and $G$ in $H_{N,a+n,k}^{pos}$. Hence, it is necessary to look carefully at the impact of these two ingredients on the approach. For instance, $G$ should be chosen to ensure  compatibility between $G$ and data. That is, to avoid the so-called ``prior-data conflict" \cite[]{Evans06}. As for $a$, we assess the effect of this parameter on $H_{N,a+n,k}^{pos}$ for fixed $n$ as $N\rightarrow\infty$ in the next theorem.
\begin{theorem}\label{a-increasing}
Let $\mathbf{x}_{d\times n}$ be a sample from $d$-variate distribution function $F$ and $F|\mathbf{x}\sim DP(a+n,G^{\ast}_{a,n})$. Then, for a fixed $n$, $\liminf E(H_{N,a+n,k}^{pos})\geq H(F)+c$ as $N\rightarrow\infty$ and $a\rightarrow\infty$, where $c\neq0$.
\end{theorem}
\begin{proof}
For fixed $a$ and $n$, similar to the proof of Theorem \ref{pos-mean}, $E(H^{pos}_{N,a+n,k})\rightarrow H(G^{\ast}_{a,n})$, as $N\rightarrow\infty$. Now, since $D_{kull}(G,G^{\ast}_{a,n})$ and $D_{kull}(F_{n},G^{\ast}_{a,n})$ are non-negative in \eqref{mix-G}, we have $H(G^{\ast}_{a,n})\geq \frac{a}{a+n}H(G)+\frac{n}{a+n}H(F_{n})=I_1+I_2$. Letting $a\rightarrow\infty$ in $I_{1}$ and $I_{2}$ gives
\begin{small}
\begin{align*}
\liminf H(G^{\ast}_{a,n})\geq H(G)=\left(H(G)-H(F)\right)+H(F)=c+H(F).
\end{align*}
\end{small}
Since the prior guess $G$ is not as the same as the true distribution $F$, then $c\neq0$ and the proof is completed.
\end{proof}

\section{Prior-based Test for Independence}\label{sec-pri.test}
Let $\mathbf{X}=(X_{1},\ldots,X_{d})$ be a random vector from an unknown distribution $F$. The problem to be addressed in this section is assessing the hypothesis
\begin{align}\label{H0}
\mathcal{H}_{0}:MI(F)=0,
\end{align}
using  BNP framework. Let $\mathbf{x}_{d\times n}$ be an observed sample of size $n$ from $F\sim DP(a,G)$. In order to implement the test, for a given choice of $a$, let $G$ be the cdf of $N(\mathbf{0}_{d},I_{d})$ and $MI^{pri}=[-H^{pri}_{N,a,k}(F)+\sum_{i=1}^{d} H^{pri}_{N,a,k}(F_i)]^{+}$ be the prior of mutual information between elements of $\mathbf{X}$. From part (i) of Lemma \ref{E-V.pri}, since $E(MI^{pri})\rightarrow MI(G)=0$ as $N\rightarrow\infty$, then $MI^{pri}$ is a good prior to compare with $MI^{pos}$ for displaying the mutual independence between $X_{1},\ldots,X_{d}$. As shown in Theorem \ref{pos-mean}, if the assumption of independence is true, the distribution of $MI^{pos}$ (posterior of mutual information) should be more concentrated around zero than the distribution of $MI^{pri}$ (i.e. $MI^{pos}$ more supports $\mathcal{H}_0$ than $MI^{pri}$); otherwise, the distribution of $MI^{pri}$ should be more concentrated at zero than the distribution of $MI^{pos}$ (i.e. $MI^{pri}$ more supports $\mathcal{H}_0$ than $MI^{pos}$). This comparison is made by using $RB$ with the interpretation as discussed in Section \ref{sec-back}. To this end, we consider an interval $[0,c)$ to compare the concentration of the distribution of the posterior to the prior. The choice of $c$ has a key role in the proposed test. As a simple tactic, we propose to fix $c$ to be close to zero (such as $c=0.05$) such that the prior probability $Pr(MI^{pri}\in [0,c))=0.5$. 
 Note that, the value of $Pr(MI^{pri}\in [0,c))$ depends on the choice of the concentration parameter $a$ in $DP(a,G)$. As Lemma \ref{E-V.pri} shows, for small values of $a$, the concentration of $MI^{pri}$ will be decreased around zero. Then, the errors of $MI^{pri}$ (i.e. large values of $MI^{pri}$) will be increased. This may cause to decrease the value of $Pr(MI^{pri}\in [0,c))$, which may lead to have  incorrect values for $RB$.
To avoid this difficulty, we need to increase the value of $a$ such that $a$ does not exceed $n/2$.  Hence, contrary to the estimation problem where $a$ should be selected to be small, the choice of $a$ in the test is different. Algorithm A helps to elicit  suitable choices for $a$ to run the test.

\begin{small}
\noindent\textbf{Algorithm A:} \textit{Selecting $a$ in $MI^{pri}$ for testing of independence}
\begin{itemize}
\item[i.] Set a small fixed value $c$, say $c=0.05$.
\item[ii.] Choose the  value of $a$ such that $Pr(0\leq MI^{pri}< c)=0.5$. The preceding probability can be estimated as follows:
\begin{itemize}
\item[a.] Generate a sample of $r$ values from $MI^{pri}$. The steps of sampling from $MI^{pri}$ are detailed in Algorithms B, Section \ref{sec-computations}.
\item[b.] Consider the ratio of the values of $MI^{pri}$ contents of $[0,c)$ as the approximation of $Pr(0\leq MI^{pri}< c)$.
\item[c.]  If the approximated probability is more
(less) than 0.5, then  decrease (increase) the value of $a$ to
reach the value of 0.5.
\end{itemize}
\end{itemize}
\end{small}

Algorithm A was thoroughly implemented for several values of $c$ and $d$.  Table \ref{prior-probability} in  Appendix B reports  appropriate values of $a$ when $d=2$.  The results for $d>2$  are found to be similar. That is, Table \ref{prior-probability} may be used for any arbitrary dimension $d$. Thus, from Table \ref{prior-probability}, an appropriate  choice of $a$ to carry out the test is  $a=1$.

\section{Computational Algorithms for Testing of Independence}\label{sec-computations}

The following algorithms summarize the main steps to carry out the test of independence for \eqref{H0}. Since closed forms of densities of $MI^{pri}$ and $MI^{pos}$ are not available, their empirical distributions are required to implement the below algorithms.

\bigskip

\noindent\textbf{Algorithm B:} \textit{Prior-based test for independence}
\begin{small}
\begin{enumerate}
\item Use Algorithm A to choose a value of $a$. Note that, $a=1$ is a recommended choice to proceed the test.

\item For the selected $a$ in the previous step, let $G$ be the cdf of $N(\mathbf{0}_{d},I_{d})$ and  generate a sample from $DP(a,G)$ as described in Section 2.1.

\item For the   sample generated in the previous step, use \eqref{k-NN.pri} to compute $H^{pri}_{N,a,k}(F)$ and $H^{pri}_{N,a,k}(F_{i})$, for $i=1,\ldots,d$.

\item
Substitute $H^{pri}_{N,a,k}(F)$ and $H^{pri}_{N,a,k}(F_{i})$'s into \eqref{MI-Entropy} to compute $MI^{pri}=[-H^{pri}_{N,a,k}(F)+\sum_{i=1}^{d}H^{pri}_{N,a,k}(F_{i})]^{+}$.

\item 
Repeat steps 2-4 to generate a sample of $\ell$ values from $MI^{pri}$.

\item
Use steps 2-5  to obtain a sample of $\ell$ values from $MI^{pos}$ by replacing $a$ by $a+n$, $G$ by $G^{\ast}_{a,n}$, \eqref{k-NN.pri} by \eqref{k-NN.post}   and  prior by posterior.

\item Let $\hat{F}_{MI^{pri}}$ denote the empirical cdf of $MI^{pri}$ based
on the prior sample in step (2). Let $\hat{F}_{MI^{pos}}$ denote the
empirical cdf of $MI^{pos}$ based on the posterior sample in step (3). Estimate $RB_{MI}(0\,|\,\mathbf{x})={\pi_{MI^{pos}}(0)}/{\pi_{MI^{pri}}(0)}$ by
\begin{equation}
\widehat{RB}_{MI}(0\,|\,\mathbf{x})=\{\hat{F}_{MI^{pos}}(c)-\hat{F}%
_{MI^{pos}}(0)\}/\{\hat{F}_{MI^{pri}}(c)-\hat{F}%
_{MI^{pri}}(0)\}, \label{rbest}%
\end{equation}
\item
Let $M$ be a positive number. For $i=0,\ldots,M,$ let $\hat{d}_{i/M}$ be the estimate of $d_{i/M},$ the $(i/M)$-th prior
quantile of $MI^{pri}$. Here $\hat{d}_{0}$ and $\hat{d}_{1}$ are, respectively, the smallest and the largest value
of the $r$ values generated in step (2). For $d\in (\hat{d}_{0},\hat{d}_{1})$, estimate the strength $DP_{MI}\big(RB_{MI}(d\,|\,\mathbf{x})\leq RB_{MI} (0\,|\,\mathbf{x})\,|\,\mathbf{x}\big)$ by the finite sum
\begin{equation}
\sum_{\{i\geq i_{0}:\widehat{RB}_{MI}(\hat{d}_{i/M}\,|\,\mathbf{x})\leq\widehat{RB}%
_{MI}(0\,|\,\mathbf{x})\}}\big(\hat{F}_{MI^{pos}}(\hat{d}_{(i+1)/M})-\hat{F}_{MI^{pos}}(\hat
{d}_{i/M})\big), \label{strest}%
\end{equation}
\end{enumerate}
\end{small}
where $i_{0}$ is chosen so that $\frac{i_{0}}{M}$ is not too small (typically $\frac{i_{0}}{M}\approx0.05)$ and $\widehat{RB}_{MI}(\hat{d}_{i/M}\,|\,\mathbf{x})=M\{\hat{F}_{MI^{pos}}(\hat{d}_{(i+1)/M})-\hat{F}_{MI^{pos}}(\hat{d}_{i/M})\}$. For fixed $M,$ as $N\rightarrow\infty$ and $\ell\rightarrow\infty,$ then $\hat{d}_{i/M}$ converges almost surely to $d_{i/M}$ and (\ref{rbest}) and (\ref{strest}) converge almost surely to $RB_{MI}(0\,|\,\mathbf{x})$ and
$DP_{MI}\big(RB_{MI}(d\,|\,\mathbf{x})\leq RB_{MI} (0\,|\,\mathbf{x})\,|\,\mathbf{x}\big)$, respectively.  The consistency of the proposed test is achieved by Proposition 6 of \cite{Al-Labadi18}.

\section{Simulation Studies}\label{simulation}
This section reveals the performance of the BNP methodology in testing of independence. To this aim, samples are generated from several $d$-variate distributions. Table \ref{notation} gives the relevant notations of these distributions. First, we consider three common $d$-variate distributions: normal, $t$-student and Maxwell-Boltzmann distributions.  We consider sample sizes $n=20,30$ and $50$. For each sample size, $r=1000$ samples were generated. Each sample gives an $RB$(strength) by setting $k=3$ in Algorithm B. For the test, we set  $\ell=1000$ in Algorithm B.  With regard to Table \ref{prior-probability}, we set $c=0.05$ and thus choose $a=1$ as outlined in Table \ref{prior-probability}. The recorded values of  $RB$ and strength (Str) are the average of the 1000 results.  For the goal of comparison, the  p-value of the test of independence \citep{Berrett19b} are reported in $r$ replication. The $\mathsf{R}$ package \textbf{IndepTest} is used to compute the p-value. The results of the BNP method and its frequentist counterpart \citep{Berrett19a,Berrett19b} are presented in Table \ref{comm.dis.T}.  It follows from Table \ref{comm.dis.T} that the prior based test has a good performance to test independence between $d$ variables. To clear up, for instance, when $N_4(\mathbf{0}_4,\Sigma_4)$ and $n=50$ in Table \ref{comm.dis.T}, the average value of relative belief ratios is $0.53$ with relevant strength $0.07$, which shows the good performance of the proposed test to reject the assumption of mutual independence.


\begin{landscape}
\begin{table}[h!]
\setlength{\aboverulesep}{0pt}
\setlength{\belowrulesep}{0pt}
\setlength{\extrarowheight}{1.1 mm}
\setlength{\tabcolsep}{6 mm}
\caption{The average values of RB(strength) for testing the mutually independent
under several distributions with $k=3$.}\label{comm.dis.T}
\scalebox{.78}
{
\begin{tabular}{cccccccccc}
\toprule
\multirow{2}[3]{*}{\bfseries Example}&$d$&\multirow{2}[3]{*}{$n$}& \multicolumn{1}{c}{BNP}&\multicolumn{1}{c}{Berrett et al.} &\multirow{2}[3]{*}{\bfseries Example}&$d$&\multirow{2}[3]{*}{$n$}&  \multicolumn{1}{c}{BNP}&\multicolumn{1}{c}{Berrett et al.} \\\cmidrule(lr){4-5}\cmidrule(lr){9-10}
&$(MI^{T})$&&$RB(Str)$ & p-value& &$(MI^{T})$&&$RB(Str)$& p-value\\
\cmidrule(lr){1-5}\cmidrule(lr){6-10}
\scalebox{1.1}{$N_{d}(\mathbf{0}_{d},I_{d})$} &\scalebox{1}{2} &\scalebox{1}{20} &$1.97(0.53)$& $0.498$ &\scalebox{1}{$t_{3}(\mathbf{0}_{d},I_{d})$}&\scalebox{1}{2}&\scalebox{1}{20}&$1.50(0.33)$&$0.366$\\
 &(0) &\scalebox{1}{30} &$2.08(0.48)$&$0.508$&&(0.042)&\scalebox{1}{30}&$1.49(0.27)$&$0.350$\\
  & & \scalebox{1}{50}&$2.11(0.52)$& $0.493$&&&\scalebox{1}{50}&$1.22(0.21)$&$0.0326$\\\cmidrule(lr){2-5}\cmidrule(lr){7-10}
   &3 &\scalebox{1}{20} &$1.99(0.56)$& $0.514$&&3&20&$1.23(0.25)$&$0.312$\\
   &(0) &\scalebox{1}{30} &$2.15(0.58)$&$0.507$&&(0.110)&30&$0.98(0.14)$&$0.289$\\
   & &\scalebox{1}{50} &$2.38(0.57)$&$0.494$&&&50&$0.81(0.09)$&$0.236$\\\cmidrule(lr){2-5}\cmidrule(lr){7-10}
   &4 &\scalebox{1}{20} &$2.02(0.62)$&$0.509$&&4&20&$1.15(0.19)$&$0.243$\\
   &(0) &\scalebox{1}{30} &$2.32(0.56)$&$0.505$&&(0.195)&30&$0.93(0.13)$&$0.218$\\
   & &\scalebox{1}{50} &$2.49(0.63)$&$0.508$&&&50&$0.58(0.05)$&$0.188$\\
\cmidrule(lr){1-5}\cmidrule(lr){6-10}
\scalebox{1.1}{$N_{d}(\mathbf{0}_{d},\Sigma_{d})$} &\scalebox{1}{2} &\scalebox{1}{20} &$1.59(0.37)$& $0.401$&\scalebox{1.1}{$t_{20}(\mathbf{0}_{d},I_{d})$}&2&20&$1.95(0.48)$&$0.482$\\
 & (0.066)&\scalebox{1}{30} &$1.53(0.34)$& $0.366$&&(0.001)&30&$2.09(0.49)$&$0.503$\\
  & &\scalebox{1}{50} &$1.45(0.21)$&$0.348$&&&50&$2.14(0.52)$&$0.486$\\\cmidrule(lr){2-5}\cmidrule(lr){7-10}
   &3 &\scalebox{1}{20} &$1.06(0.12)$& $0.269$&&3&20&$1.91(0.42)$&$0.490$\\
   & (0.235)&\scalebox{1}{30} &$0.91(0.09)$&$0.228$&&(0.003)&30&$2.10(0.48)$&$0.475$\\
   & &\scalebox{1}{50} &$0.78(0.05)$& $0.164$&&&50&$2.19(0.53)$&$0.474$\\\cmidrule(lr){2-5}\cmidrule(lr){7-10}
   &4 &\scalebox{1}{20} &$0.98(0.13)$&$0.230$&&4&20&$1.98(0.51)$&$0.478$\\
   &(0.450) &\scalebox{1}{30} &$0.72(0.09)$&$0.160$&&(0.006)&30&$2.13(0.54)$&$0.468$\\
   & &\scalebox{1}{50} &$0.51(0.05)$&$0.109$&&&50&$2.31(0.62)$&$0.471$\\
\cmidrule(lr){1-5}\cmidrule(lr){6-10}
\scalebox{1.1}{$N_{d}(\mathbf{0}_{d},A_{d})$} &\scalebox{1}{2} &\scalebox{1}{20} &$0.91(0.09)$&$0.297$&\scalebox{1.1}{$Mwell(\mathbf{10}_{d})$}&2&20&$1.89(0.52)$&$0.500$\\
 &(0.143) &\scalebox{1}{30} &$0.63(0.06)$&$0.267$&&(0)&30&$2.13(0.47)$&$0.503$\\
  & &\scalebox{1}{50} &$0.56(0.04)$&$0.214$&&&50&$2.15(0.50)$&$0.507$\\\cmidrule(lr){2-5}\cmidrule(lr){7-10}
   &3 &\scalebox{1}{20} &$1.27(0.21)$& $0.312$&&3&20&$1.94(0.54)$&$0.520$\\
   & (0.143)&\scalebox{1}{30} &$1.06(0.10)$&$0.265$&&(0)&30&$2.17(0.64)$&$0.518$\\
   & &\scalebox{1}{50} &$0.64(0.09)$& $0.198$&&&50&$2.24(0.59)$&$0.489$\\\cmidrule(lr){2-5}\cmidrule(lr){7-10}
   &4 &\scalebox{1}{20} &$1.85(0.31)$&$0.317$&&4&20&$2.16(0.64)$&$0.499$\\
   &(0.143) &\scalebox{1}{30} &$1.71(0.39)$& $0.275$&&(0)&30&$2.19(0.67)$&$0.503$\\
   & &\scalebox{1}{50} &$0.91(0.18)$& $0.220$&&&50&$2.50(0.70)$&$0.492$\\
\bottomrule
\end{tabular}
}\label{prior-data}
\end{table}
\end{landscape}

Similar to the study of \cite{Roy}, to consider more interesting scenarios, we included the following six unusual bivariate distributions (UBD):

\noindent\emph{Four clouds}: Let $Z_{1},Z_2,T_1$ and $T_2$ be independent with $Z_1,Z_2\sim N_1(0,1)$ and $Pr(T_1=\pm1)=Pr(T_2=\pm1)=1/2$. Then, consider the random vector $(X_1,X_2)$ with $X_1=Z_1+T_1$ and $X_2=Z_2+T_2$.

\noindent\emph{Circle}: Let $Z_{1},Z_2$ and $U$ be independent with $Z_1,Z_2\sim N_1(0,1)$ and $U\sim U(-1,1)$. Then, consider the random vector $(X_1,X_2)$ with $X_1=\sin(\pi U)+Z_{1}/8 $ and $X_2=\cos(\pi U)+Z_{2}/8 $.

\noindent\emph{Two Parabolas}: Let $U_{1},U_2$ and $T$ be independent with $U_1\sim U(-1,1)$, $U_2\sim U(0,1)$ and $Pr(T=\pm1)=1/2$. Then, consider the random vector $(X_1,X_2)$ with $X_1=U_1 $ and $X_2=T\big(U_{1}^{2}+U_{2}/2 \big)$.

\noindent\emph{Parabola}: Let $U_{1}$ and $U_2$ be independent with $U_1\sim U(-1,1)$ and $U_2\sim U(0,1)$. Then, consider the random vector $(X_1,X_2)$ with $X_1=U_1 $ and $X_2=\big(U_{1}^{2}+U_{2}/2 \big)/2 $.

\noindent\emph{Diamond}: Let $U_{1},U_2\overset{i.i.d.}{\sim}U(-1,1)$. Then, consider the random vector $(X_1,X_2)$ with $X_1=U_1\cos(-\pi/4)+U_2\sin(-\pi/4)$ and $X_2=-U_1\sin(-\pi/4)+U_2\cos(-\pi/4)$.

\noindent\emph{W}:  Let $U_{1}$ and $U_2$ be independent with $U_1\sim U(-1,1)$ and $U_2\sim U(0,1)$. Then, consider the random vector $(X_1,X_2)$ with $X_1=U_1+U_{2}/3 $ and $X_2=4\big( (U_{1}^{2}-1/2)^{2}+U_{2}/n \big)$.

Figure \ref{UBD-plot} shows plots of samples generated from the above distributions. The interesting property of these distributions is that in each pair of random variables, $X_1$ and $X_2$ are uncorrelated but dependent, except in \emph{four clouds} where $X_1$ and $X_2$ are uncorrelated and independent. 
Table \ref{unusual-table} shows that the assumption of mutual independence is accepted only for \emph{four clouds} in the cases where the  sample size is greater than or equal to $30$.

\begin{figure}[ht]
 \centering
    \subfloat{{\includegraphics[width=4.5cm]{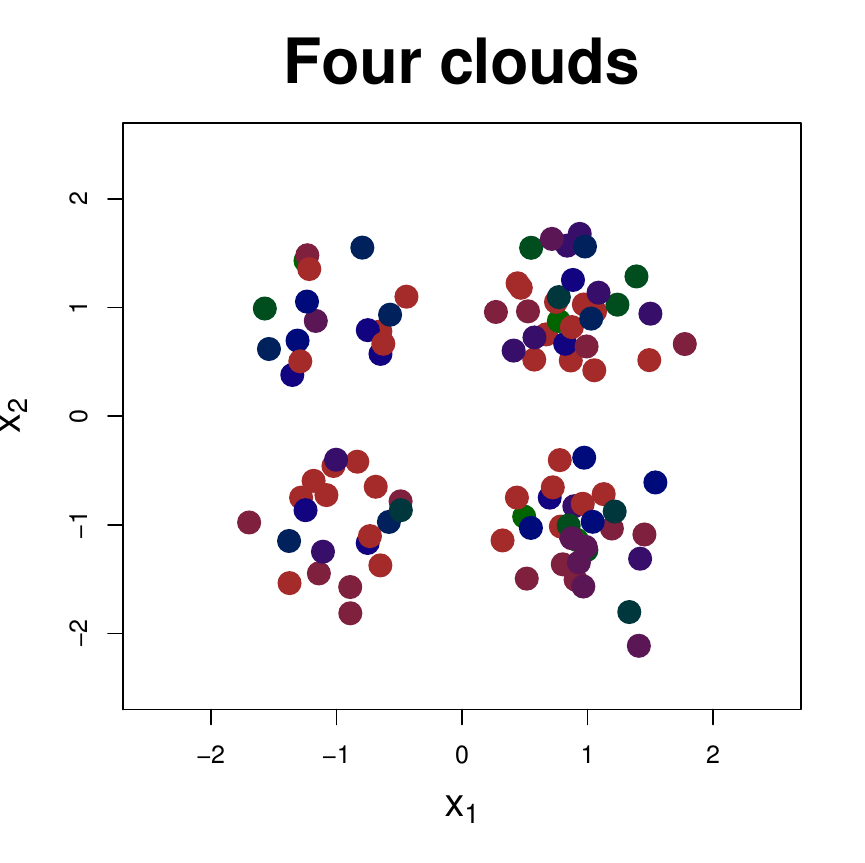} }}%
    \subfloat{{\includegraphics[width=4.5cm]{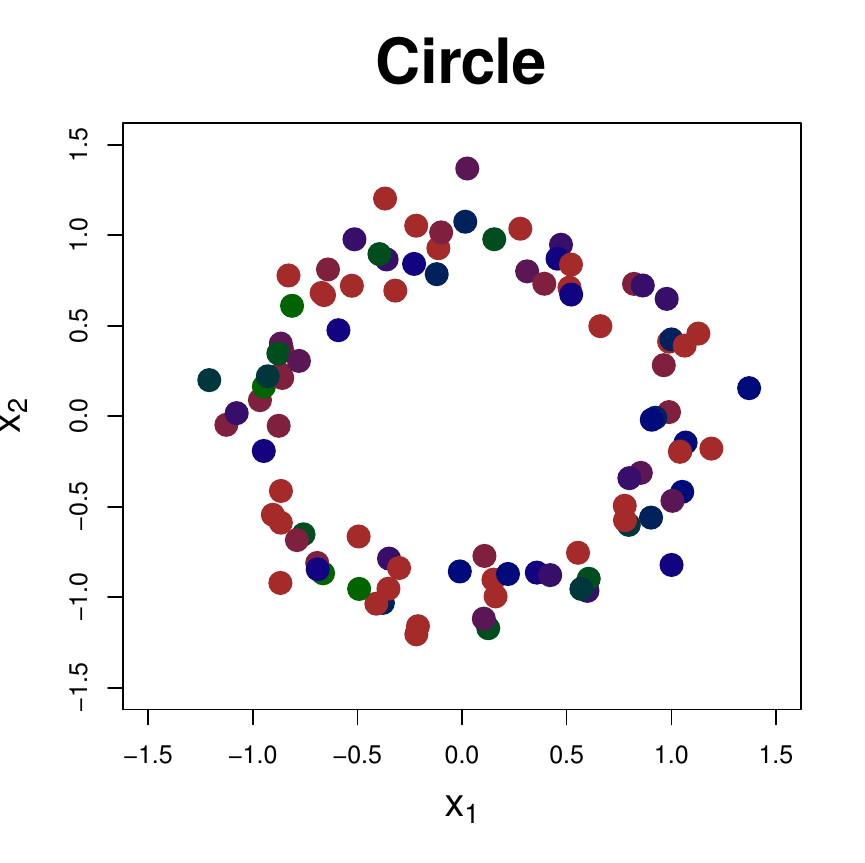} }}
    \subfloat{{\includegraphics[width=4.5cm]{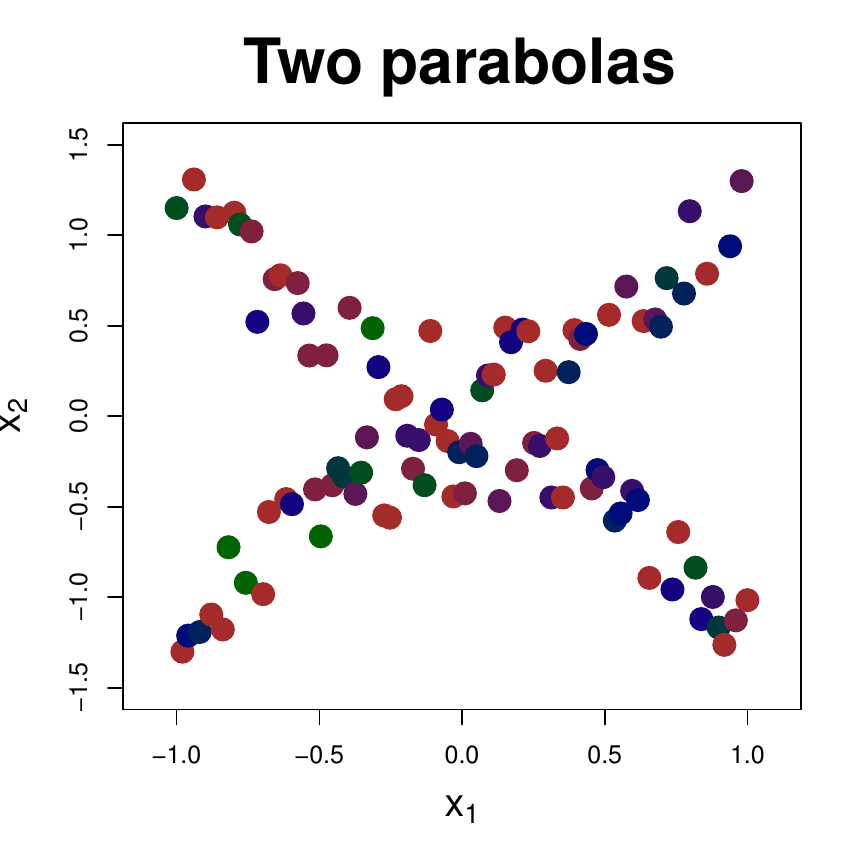} }}
    \qquad
    \subfloat{{\includegraphics[width=4.5cm]{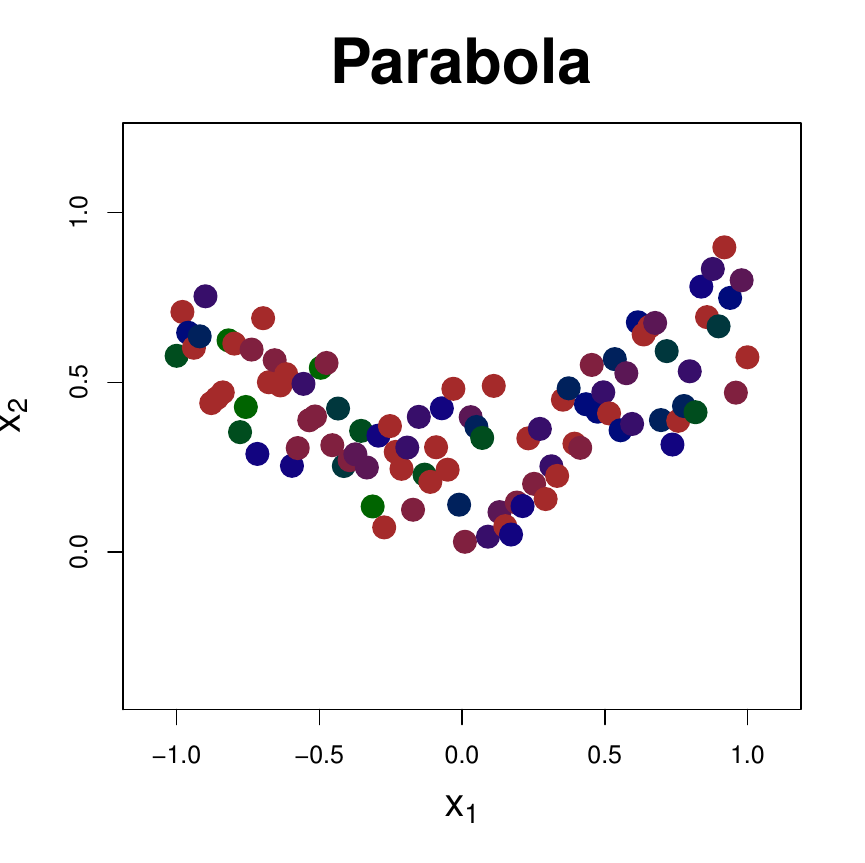} }}%
    \subfloat{{\includegraphics[width=4.5cm]{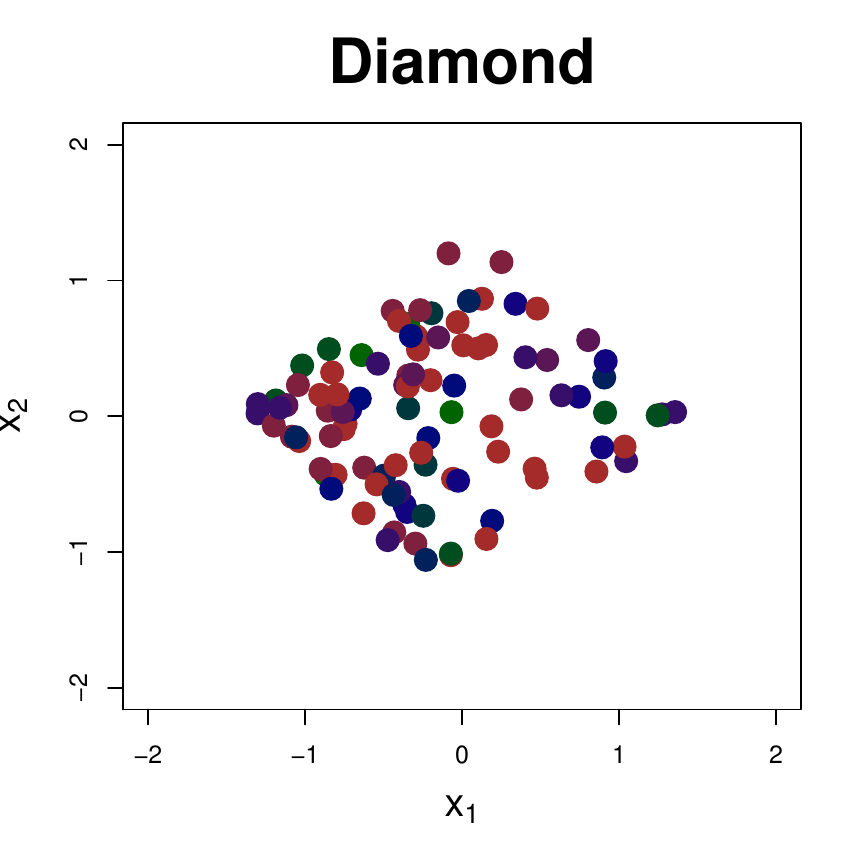} }}
    \subfloat{{\includegraphics[width=4.5cm]{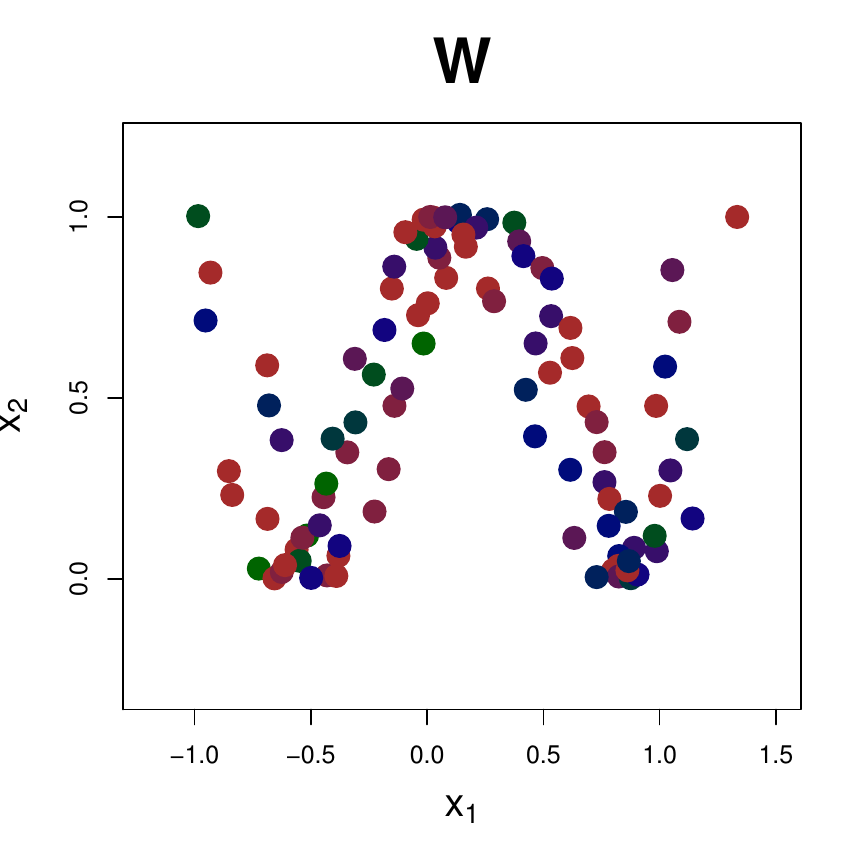} }}
    \caption{Samples generated from six UBDs with sample size of $n=100$.}\label{UBD-plot}%
\end{figure}
\begin{table}
\setlength{\tabcolsep}{.17 cm}
\caption{The average values of the $RB$ and its $Str$ over $r$ samples generated from six UBDs with $k=3$.}
\begin{tabular}{cccccccccc}
\toprule
\multirow{2}[3]{*}{\bfseries UBD}&\multirow{2}[3]{*}{$n$}& \multicolumn{2}{c}{BNP}&\multicolumn{1}{c}{Berrett et al.}&\multirow{2}[3]{*}{\bfseries UBD}&\multirow{2}[3]{*}{$n$}& \multicolumn{2}{c}{BNP}&\multicolumn{1}{c}{Berrett et al.} \\\cline{3-4}\cline{8-9}
&&$RB$&$Str$ & $p$-value& &&$RB$&$Str$ & $p$-value\\
\hline
Four clouds&20&$1.67$&$0.40$&$0.493$& Parabola&20&$0.88$&$0.26$&$0.061$\\
&30&$2.17$&$0.48$&$0.498$& &30&$0.66$&$0.12$&$0.036$\\
&50&$2.25$&$0.57$&$0.499$& &50&$0.46$&$0.09$&$0.009$\\
\cline{1-5}\cline{6-10}
Circle&20&$0.97$&$0.21$&$0.178$& Diamond&20&$1.07$&$0.23$&$0.346$\\
&30&$0.60$&$0.08$&$0.049$& &30&$0.89$&$0.20$&$0.273$\\
&50&$0.19$&$0.00$&$0.009$& &50&$0.77$&$0.12$&$0.175$\\
\cline{1-5}\cline{6-10}
Two parabolas&20&$0.74$&$0.15$&$0.045$& W&20&$1.71$&$0.49$&$0.174$\\
&30&$0.31$&$0.03$&$0.012$& &30&$0.99$&$0.22$&$0.064$\\
&50&$0.13$&$0.01$&$0.009$& &50&$0.62$&$0.02$&$0.009$\\
\bottomrule
\end{tabular}\label{unusual-table}
\end{table}

Finally, to evaluate the performance of the proposed method on a real data set, the combined cycle power plant (CCPP) data set is considered. This data set contains 9568 five-dimensional data points. It is collected from 2006 to 2011 and is available at \url{https://archive.ics.uci.edu/ml/datasets/combined+cycle+power+plant}. Its goal is to predict the net hourly electrical energy output of the plant based on the temperature (T), the ambient pressure (AP), the relative humidity (RH) and the exhaust vacuum (V). Thus, it is significant to check whether the four variables T, AP, RH, and V are independent. In addition, besides using all 9568 data points, we considered three samples with sample sizes $n=20,30$ and $50$ generated randomly from the whole data set. The proposed method then is implemented. The results are reported in Table \ref{CCPP}, where it follows clearly from this table  that the assumption of independence between T, AP, RH and V is rejected in all cases. 

%

\begin{table}[h!]
\center
\setlength{\aboverulesep}{0pt}
\setlength{\belowrulesep}{0pt}
\setlength{\extrarowheight}{1 mm}
\setlength{\tabcolsep}{10.4 mm}
\caption{ The result of the BNP test ($RB$ and strength) and the p-value of the test of \cite{Berrett19b} for CCPP data set with $k=3$ and various sample sizes $n$.}\label{CCPP}
\scalebox{0.8}
{
\begin{tabular}{ccc}
\toprule
\smallskip
\multirow{2}[3]{*}{\bfseries CCPP}& \multirow{2}[3]{*}{$RB(Str)$}&\multirow{2}[3]{*}{p-value} \\

$n$&  & \\
\hline
$20$&$0.65(0.07)$&$0.237$\\
$30$&$0.60(0.02)$&$0.019$\\
$50$&$0.29(0.01)$&$0.009$\\
$9568$&$0.10(0.00)$&$0.009$\\
\bottomrule
\end{tabular}
}
\end{table}

\section{Concluding Remarks}\label{conclusion}
The BNP prior and posterior of mutual information have been proposed. They have been constructed based on using the Dirichlet process  and the $k$-nearest neighbor distance. Several interesting theoretical results have been presented. As a result, a new Bayesian test of independence has been developed. The performance of the procedure has been examined by several interesting examples. The obtained results reflect the excellent performance of the methodology in testing.

\begin{appendices}
\section{Proof of Equation \eqref{g}, \eqref{fn} and \eqref{g-ast}}
We show \eqref{g} by two below steps (the proof for \eqref{fn} and \eqref{g-ast} are similar).

\noindent Step 1: For $d=3$, consider $\mathbf{Z}=(Z_1,Z_2,Z_3)\sim G$ and $\left\lbrace\mathbf{t}_{j}\right\rbrace_{j\in\mathcal{I}}\subseteq\mathbb{R}^{3}$, where $\mathbf{t}_{j}=(t_{j1},t_{j2},t_{j3})$. Then for a given $\delta>0$, we can write
\begin{small}
\begin{align}\label{I's}
Pr(Z_{1}\leq t_{j1}, Z_{2}\leq t_{j2},Z_{3}\leq t_{j3})&=Pr(t_{j1}-\delta<Z_{1}\leq t_{j1}, t_{j2}-\delta<Z_{2}\leq t_{j2},t_{j3}-\delta<Z_{3}\leq t_{j3})\nonumber\\
&+Pr(Z_{1}\leq t_{j1}-\delta, t_{j2}-\delta<Z_{2}\leq t_{j2},t_{j3}-\delta<Z_{3}\leq t_{j3})\nonumber\\
&+Pr(t_{j1}-\delta<Z_{1}\leq t_{j1},Z_{2}\leq t_{j2}-\delta,t_{j3}-\delta<Z_{3}\leq t_{j3})\nonumber\\
&+Pr(t_{j1}-\delta<Z_{1}\leq t_{j1}, t_{j2}-\delta<Z_{2}\leq t_{j2},Z_{3}\leq t_{j3}-\delta)\nonumber\\
&+Pr(Z_{1}\leq t_{j1}-\delta, Z_{2}\leq t_{j2}-\delta,t_{j3}-\delta<Z_{3}\leq t_{j3})\nonumber\\
&+Pr(Z_{1}\leq t_{j1}-\delta, t_{j2}-\delta<Z_{2}\leq t_{j2},Z_{3}\leq t_{j3}-\delta)\nonumber\\
&+Pr(t_{j1}-\delta<Z_{1}\leq t_{j1}, Z_{2}\leq t_{j2}-\delta,Z_{3}\leq t_{j3}-\delta)\nonumber\\
&+Pr(Z_{1}\leq t_{j1}-\delta, Z_{2}\leq t_{j2}-\delta,Z_{3}\leq t_{j3}-\delta)\nonumber\\
&=I_1+I_2+I_3+I_4+I_5+I_6+I_7+I_8.
\end{align}
\end{small}
\noindent On the other hand,
\begin{small}
\begin{align*}
I_2+Pr(Z_{1}\leq t_{j1}-\delta, Z_{2}\leq t_{j2}-\delta,Z_{3}\leq t_{j3}-\delta)=Pr(Z_{1}\leq t_{j1}-\delta, Z_{2}\leq t_{j2},Z_{3}\leq t_{j3}).
\end{align*}
\end{small}
Then, we have
\begin{small}
\begin{align}\label{I-2}
I_2=G(t_{j1}-\delta,  t_{j2}, t_{j3})-G(t_{j1}-\delta, t_{j2}-\delta, t_{j3}-\delta).
\end{align}
\end{small}
\noindent Similarly,
\begin{small}
\begin{align}
I_3&=G(t_{j1},  t_{j2}-\delta, t_{j3})-G(t_{j1}-\delta, t_{j2}-\delta, t_{j3}-\delta),\label{I-3}\\
I_4&=G(t_{j1},  t_{j2}, t_{j3}-\delta)-G(t_{j1}-\delta, t_{j2}-\delta, t_{j3}-\delta),\label{I-4}\\
I_5&=G(t_{j1}-\delta,  t_{j2}-\delta, t_{j3})-G(t_{j1}-\delta, t_{j2}-\delta, t_{j3}-\delta),\label{I-5}\\
I_6&=G(t_{j1}-\delta,  t_{j2}, t_{j3}-\delta)-G(t_{j1}-\delta, t_{j2}-\delta, t_{j3}-\delta),\label{I-6}\\
I_7&=G(t_{j1},  t_{j2}-\delta, t_{j3}-\delta)-G(t_{j1}-\delta, t_{j2}-\delta, t_{j3}-\delta).\label{I-7}.
\end{align}
\end{small}
by substituting \eqref{I-2}, \eqref{I-3}, \eqref{I-4}, \eqref{I-5}, \eqref{I-6}, and \eqref{I-7} into \eqref{I's}, we have
\begin{small}
\begin{align*}
G(t_{j1},  t_{j2}, t_{j3})&=Pr(t_{j1}-\delta<Z_{1}\leq t_{j1}, t_{j2}-\delta<Z_{2}\leq t_{j2},t_{j3}-\delta<Z_{3}\leq t_{j3})\nonumber\\
&+G(t_{j1}-\delta,  t_{j2}, t_{j3})+G(t_{j1},  t_{j2}-\delta, t_{j3})+G(t_{j1},  t_{j2}, t_{j3}-\delta)\nonumber\\
&+G(t_{j1}-\delta,  t_{j2}-\delta, t_{j3})+G(t_{j1}-\delta,  t_{j2}, t_{j3}-\delta)+G(t_{j1},  t_{j2}-\delta, t_{j3}-\delta)\nonumber\\
&+\left(-(2^3-2)+1\right)G(t_{j1},  t_{j2}, t_{j3}).
\end{align*}
\end{small}
After simplification, we get
\begin{small}
\begin{align*}
Pr(t_{j1}-\delta<Z_{1}\leq t_{j1}, t_{j2}-\delta<Z_{2}\leq t_{j2},t_{j3}-\delta<Z_{3}\leq t_{j3})&=G(t_{j1},  t_{j2}, t_{j3})\nonumber\\
&+(2^3-3)G(t_{j1}-\delta,  t_{j2}-\delta, t_{j3}-\delta)\nonumber\\
&-\underset{S_3}{\sum} G(s_{1},s_{2},s_{3}),
\end{align*}
\end{small}
where $S_3=\lbrace (s_1,s_2,s_3):\, s_k\in\lbrace t_{jk}-\delta,t_{jk}\rbrace, k\in\lbrace 1,2,3\rbrace\rbrace\setminus\lbrace(t_{j1}-\delta,  t_{j2}-\delta, t_{j3}-\delta),(t_{j1},  t_{j2}, t_{j3})\rbrace$.

\noindent Step 2: Now, generalize step 1 for $d>3$ to conclude the result.

\newpage

\section{}
\begin{small}
\begin{table}[h!]
\center
\setlength{\aboverulesep}{0pt}
\setlength{\belowrulesep}{0pt}
\setlength{\extrarowheight}{1.1 mm}
\setlength{\tabcolsep}{9.5 mm}
\caption{Values of $Pr(MI^{pri}\in[0,c))$ in Algorithm A to choose $a$ for the BNP test of independence with $k=3$ and $d=2$.}
\scalebox{.9}
{
\begin{tabular}{c|cccc}
\toprule
$c$&$a=0.05$&\cellcolor{gray!10}$a=1$&$a=5$&$a=10$\\
\hline
$0.01$&0.314&\cellcolor{gray!10}0.473&0.477&0.490\\
$0.02$&0.317&\cellcolor{gray!10}0.479&0.486&0.499\\
$0.03$&0.323&\cellcolor{gray!10}0.481&0.495&0.510\\
$0.04$&0.325&\cellcolor{gray!10}0.491&0.510&0.520\\
\cellcolor{gray!10}$0.05$&\cellcolor{gray!10}0.327&\cellcolor{gray!30}0.498&0.521&0.538\\
$0.06$&0.331&0.516&0.533&0.549\\
$0.07$&0.333&0.520&0.549&0.563\\
$0.08$&0.337&0.533&0.554&0.579\\
$0.09$&0.342&0.548&0.576&0.601\\
$0.1$&0.344&0.568&0.600&0.626\\
\bottomrule
\end{tabular}
}\label{prior-probability}
\end{table}
\end{small}
\begin{small}
\setlength{\extrarowheight}{2mm}
\begin{table}[h!]
\centering
\caption{Description of notations}\label{notation}
\scalebox{0.8}{
\begin{tabular}{l}
\toprule
{\parbox{16cm}{1. $\mathbf{c}_{2}:=(c,c)^{T}$, $I_{2}:=\bigl[\begin{smallmatrix} 1&0\\0&1\end{smallmatrix}\bigr]$, $A_{2}:=\bigl[\begin{smallmatrix} 1&0.5\\0.5&1\end{smallmatrix}\bigr]$, $A_{3}:=\Big[\begin{smallmatrix} 1&0&0\\0&1&0.5\\0&0.5&1\end{smallmatrix}\Bigr]$, $A_{4}:=\biggl[\begin{smallmatrix} 1&0&0&0\\0&1&0&0\\0&0&1&0.5\\0&0&0.5&1\end{smallmatrix}\biggr]$, $\Sigma_{4}:=\biggl[\begin{smallmatrix} 1&0.5&0.5&0.5\\0.5&2&0.5&0.5\\0.5&0.5&1&0.5\\0.5&0.5&0.5&1\end{smallmatrix}\biggr]$, and $B_{3}:=\Bigl[\begin{smallmatrix} 1&0.9&0.9\\0.9&1&0.9\\0.9&0.9&1\end{smallmatrix}\Bigr]$.}}\\
{\parbox{16.1cm}{2. $U(a,b)$: An univariate uniform distribution with parameters $a$ and $b$.}}\\
{\parbox{16cm}{3. $F_{1}\otimes\ldots\otimes F_{d}$: A $d$-variate distribution with $d$ independent marginal distributions $F_{1},\ldots,F_{d}$.}}\\
{\parbox{16cm}{4. $Mwell(\mathbf{c}_{d})=Mwell(c)\otimes\cdots\otimes Mwell(c)$, where $Mwell(c)^\dagger$ denotes the Maxwell-Boltzman distribution with scale parameter $c$ and $MI^T=0$.}}\\
{\parbox{16cm}{5. $N_{d}(\mathbf{0}_{d},\Sigma_{d})$: A $d$-variate normal distribution with mean vector $\mathbf{0}_{d}$ and covariance matrix $\Sigma_{d}$, and $MI^T=\frac{d}{2}\log(2\pi e\sigma^{2}_{i})-\frac{1}{2}\log((2\pi e)^{d}\det(\Sigma))$, where $\sigma^{2}_{i}$ is the $i$-th diagonal element of $\Sigma_d$.}}\\
{\parbox{16cm}{6. $t_{r}(\mathbf{0}_{d},I_{d})^{\dagger}$: A $d$-variate $t$-student distribution with location parameter $\mathbf{0}_{d}$, scale parameter $I_{d}$ and $r$ degrees of freedom, and $MI^T=d\big(\frac{r+1}{2}[\psi((1+r)/2)-\psi(r/2)]+\log[\sqrt{r}B(r/2,1/2)]\big)-\big\lbrace-\log\frac{\Gamma((r+d)/2)}{\Gamma(r/2)(r\pi)^{d/2}}+\frac{r+d}{2}[\psi(\frac{r+d}{2})-\psi(\frac{r}{2})]\big\rbrace$, where $B(\cdot,\cdot)$ denotes beta function.}}\\
{\parbox{16.1cm}{7. $SP_d(LN(0,0.25))^{\dagger}$: A $d$-variate spherical distribution with lognormal distribution $LN(0,0.25)$ for radii.}}\\
\bottomrule
\end{tabular}
}
\centering
\begin{tablenotes}
      \item \small{$^{\dagger}$ Required $\mathsf{R}$ packages: shotGroups and distrEllipse.}
   \end{tablenotes}
    \end{table}
    \end{small}
\end{appendices}

\end{document}